\documentclass[a4paper,UKenglish,cleveref, autoref, thm-restate]{lipics-v2021}

\usepackage[utf8]{inputenc}
\usepackage[T1]{fontenc}
\usepackage[colorinlistoftodos]{todonotes}
\usepackage{float}
\usepackage{babel}
\usepackage{tikz}
\usepackage{verbatim}
\usepackage{cite}
\usetikzlibrary{fit}
\usetikzlibrary{arrows,shapes}
\usepackage{amsmath}

\usepackage{thmtools}
\usepackage{relsize}

\usepackage{tikz}
\usepackage{tikz-qtree}
\usetikzlibrary{trees}
\usetikzlibrary{decorations.pathmorphing}
\usetikzlibrary{decorations.markings}
\usetikzlibrary{decorations.pathmorphing,shapes}
\usetikzlibrary{calc,decorations.pathmorphing,shapes}
\usetikzlibrary{patterns}

\newtheorem{problem}{Problem}
\newcounter{proofcount}
\crefname{proofcount}{Appendix Proof}{Proof}

\newenvironment{proofcount}[1][\proofname]{
\refstepcounter{proofcount}
\par
  \pushQED{\qed}%
  \trivlist
  \item[\hskip\labelsep
        \color{lipicsGray}\sffamily\bfseries
    #1.]\ignorespaces
}{%
  \popQED\endtrivlist
}

\title{Explicit Combinatoric Structures of Palindromes and Chromatic Number of Restriction Graphs}

\author{Amihood Amir}{Department of Computer Science, Bar Ilan University, Ramat Gan 52900, Israel; and Georgia Tech, College of Computing}{amir@cs.biu.ac.il}{[orcid]}{Partially supported by ISF grant 168/23 and BSF grant 2018-141.}

\authorrunning{Amihood Amir and Michael Itzhaki}

\author{Michael Itzhaki}{Department of Computer Science, Bar Ilan University, Ramat Gan 52900, Israel}{michaelitzhaki@gmail.com}{[orcid]}{Partially supported by ISF grant 168/23.}
\Copyright{Amihood Amir and Michael Itzhaki}

\begin{document}
\nolinenumbers
%
%
%
%

\titlerunning{Combinatorial Structures of Palindromes}
\maketitle              

\begin{abstract}
The palindromic fingerprint of a string $S[1\ldots n]$ is the set $PF(S) = \{(i,j)~|~ S[i\ldots j] \textit{ is a maximal }\\ \textit{palindrome substring of } S\}$. In this work, we consider the problem of string reconstruction from a palindromic fingerprint. That is, given an input set of pairs $PF \subseteq [1\ldots n] \times [1\ldots n]$ for an integer $n$, we wish to determine if $PF$ is a valid palindromic fingerprint for a string $S$, and if it is, output a string $S$ such that $PF= PF(S)$. I et al. [SPIRE2010] showed a linear reconstruction algorithm from a palindromic fingerprint that outputs the lexicographically smallest string over a minimum alphabet. They also presented an upper bound of $\mathcal{O}(\log(n))$ for the maximal number of characters in the minimal alphabet.

In this paper, we show tight combinatorial bounds for the palindromic fingerprint reconstruction problem. We present the string $S_k$, which is the shortest string whose fingerprint $PF(S_k)$ cannot be reconstructed using less than $k$ characters. The results additionally solve an open problem presented by I et al.

\end{abstract}

\newpage

\section{Introduction}\label{s:intro}

One of the methods for understanding the essence of data structures
is by analyzing a candidate to decide whether it is a legal instance
of the data structure. Suppose, for example, one is given a 
permutation $P$ and needs to decide if a string $S$ exists such that $P$ is
its {\em suffix array}. Another example, given a matrix, decide if it is an adjacency matrix of a planar graph. We henceforth refer to this problem as
{\em recognizing} or {\em reconstructing} a data structure.

There are several reasons for wanting to recognize data structures.
1) From a purely theoretical point of view, there is a clear
  theoretical combinatorial interest in this endeavor. 2) From a practical point of view, general graph problems are
  known to be hard~\cite{GJ-79}. Nevertheless, a subset of these problems are
efficiently computable on special graphs, such as {\em planar graphs},
{\em interval graphs}, etc. Quickly recognizing these special graphs
can provide us with efficient algorithms in specialized cases. 3) Understanding the nature of data structures may lead to more efficient ways of constructing them. 4) Understanding the nature of a data structure may aid in ways of efficiently maintaining it in a dynamic environment. 

The intriguing line of research into efficiently recognizing data
structures has a long history in pattern matching and is an ongoing topic of interest ~\cite{AKL:23,AKLMS:24,DLL:05,FGLRSSY:00,IIBT:14,KTV:13,DLL:09,BIST:03,CCR:09,KPP:17,NOIIBT:17,GI:19,abik:20,MAIN1984422}.

Syntactic regularities in strings play a pervasive role in many facets of data analysis.
Searching for repeated patterns, periodicities, symmetries, and other similar forms or unusual
patterns in objects is a recurrent task in the compression of data, symbolic dynamics, genome
studies, intrusion detection, and countless other activities. A palindrome is a string that is equal to its reverse. Palindromes are a basic syntactic regularity that has been explored for
thousands of years. In addition to its myriad theoretical virtues, the palindrome also
plays an important role in nature. Because the DNA is double-stranded,
its base-pair representation offers palindromes in hairpin structures,
for example. A variety of restriction enzymes recognize and cut specific
palindromic sequences. In addition, palindromic sequences play roles in
methyl group attachment and T cell receptors. For examples of
the varied roles of palindromes in biology see,
e.g.~\cite{gk:97,f:03,lssnz:05,sr:12}.

 Consider, now, the problem of reconstructing a string from its {\em palindromic fingerprint}. We formally define the problem as follows. 
\begin{definition}[Palindromic Fingerprint]\label{d:palf}
Let $S$ be a string and $(i, j)$ be a palindromic descriptor of $S$ using start-end representation, i.e., $S[i..j]$ is a palindrome. $(i, j)$ is a {\em maximal palindrome} if and only it cannot be extended, meaning that the palindromic descriptor $(i-1,j+1)$ is either not defined, or is not a palindrome. We refer to a nonempty maximal palindrome as a {\em palindromic substring}.

The {\em palindromic fingerprint} of $S$, or simply a {\em Fingerprint}, is a list of all pairs of indices $(i,j)$, where $S[i\ldots j]$ is a maximal palindrome. We denote this fingerprint by $PF(S)$.

\end{definition}

The problem of {\em reconstructing a string from its palindromic fingerprint} is:
\begin{problem}[The palindromic fingerprint reconstruction problem]\label{p:main}~\\
{\bf Input}: A natural number $n$ and a list $F$ of pairs of indices $(i,j)$, where $i<j$ and $1 \leq i,j \leq n$.\\
{\bf Output:} A string $S$ of length $n$ where $PF(S)=F$, or report that there can not exist a string $S$ of length $n$ whose palindromic fingerprint is $F$.
\end{problem}
An example can be found at~\cref{e:p1}.

\subsection{Related Work}
I et al.~\cite{tomohiro:10} talk in their paper about problem~\ref{p:main} and solve it in linear time with the minimum alphabet by using a greedy approach that introduces the lexicographically smallest character to the reconstruction on demand. To the best of our knowledge, the paper by I et al. talks about a lexicographically minimal string but does not provide proof that their algorithm produces the minimal alphabet size. In this paper we prove that their algorithm indeed produces a minimal alphabet size.

Gawrychowski et al.~\cite{GKRR:20} present a general linear-time algorithm for reconstructing a string 
satisfying a set of positive-only constraints, in particular from runs or palindromes. A positive-only 
constraint is a constraint of the form $T[i\ldots i+\ell] = T[j\ldots j+\ell]$. Such a set of constraints can 
present palindromes, but cannot describe maximal palindromes, as maximal palindromes also introduce a negative constraint. The general 
reconstruction problem over the minimal alphabet, using both negative and positive constraints on a string is known
to be $\mathcal{NP}$-hard (Shown at~\cref{c:colnp}), as it is equivalent to the graph-coloring problem, even when the negative 
constraint length is at most 1, i.e., only constraints of the form $a[i] \neq a[j]$. It is well known that calculating the chromatic number of a general graph is hard~\cite{GJS:74}. 

\subsection{Our Contribution}
In this paper, we discuss problem~\ref{p:main}. We present an explicit structure of a fingerprint that needs at least $k$ characters to reconstruct and provide a method to reconstruct any palindromic fingerprint with exactly $k$ different characters, where $k$ is an algorithm's argument. By doing so, we solve the open problem presented by I et al.~\cite{tomohiro:10}, of finding a string with a fingerprint that contains exactly $k$ different alphabet symbols.

Our main results show a careful analysis of palindromic fingerprints - tight bounds on the relation between the fingerprint size and the minimal alphabet size of a string reconstructing it. We also analyze different types of palindromes and show an explicit structure of the shortest palindromic fingerprint that requires $k$ characters to reconstruct. We also present problem~\ref{p:main} as a graph coloring problem, which, combined with the algorithm of I et al., can produce all possible alphabet sizes for a given palindromic fingerprint.

Our main result implies the following theorem:
\begin{restatable}{theorem}{maintheorem}\label{t:main}
Let $IPF(k)$ be the shortest length $n$ of a palindromic fingerprint that can not be reconstructed with less than $k$ characters. The function is the following:

$$
IPF(k)=\begin{cases}
			1, & \text{if $k=1$}\\
            2^{k-2}+1, & \text{otherwise.}
		 \end{cases}
$$
\end{restatable}
A corollary of the theorem is that any fingerprint $F$ of length $n$ can be reconstructed with an alphabet of size at most $\log(n-1)+2$.

\section{Preliminaries}\label{s:pre}

Let $S = S[1]S[2] \ldots S[n]$ be a string. We denote the length of $S$ as $|S| = n$. The empty string with $|S| =0$ is denoted as $\varepsilon$. For two integers $i, j \in \{1\ldots n\}$, we denote as $S[i\ldots j]$ the substring of $S$ starting at index $i$ and ending at index $j$, i.e., $S[i\ldots j] = S[i]S[i+1]\ldots S[j]$. If $j < i$, $S[i\ldots j]$ is the empty string $\varepsilon$. A prefix of $S$ is a substring $S[1\ldots i]$ starting at index $1$, and a suffix of $S$ is a substring $S[i\ldots n]$ ending at index $n$.

For two strings $S$ and $T$, we denote the concatenation of $S$ and $T$ as \\
$ST = S[1]S[2] \ldots S[|S|] T[1] T[2] \ldots T[|T|]$. 
For an integer $i \in \mathbb{N}$, we denote as $S^i$ the concatenations of a string to itself $i$ times, i.e., $S^1 = S$, and $S^i = S^{i-1} \cdot S$, for any $i\geq 2$.

$S$ is called a \textit{palindrome} or \textit{palindromic} if $\forall_{ 1 \le i < |S|} S[i]=S[|S|-i+1]$. For example, $S=abcba$, or $S=abccba$ are palindromes. However, $S=abcbaa$ is not a palindrome. A nonempty substring $P=S[i\ldots j]$ is called a maximal palindrome of $S$ if $P$ is a palindrome, and either $i=1$, $j=n$, or $S[i-1] \neq S[j+1]$. The center of $P$ is $c_P = \frac{i+j}{2}$ and its radius is $r_P=\frac{j-i}{2}$. 
Note that there are $2n-1$ centers in $S$, each of which is the center of at most one maximal palindrome substring. \footnote{When counting maximal palindromes of length 0, there is exactly one palindrome at each center.}
 
\begin{definition}
    Let $S$ be a string. We say that a character $c$ {\em occurs} in $S$ if and only if $\exists i$ s.t. $S[i] = c$.
\end{definition}

\begin{definition}[Trivial palindrome]
    A {\em trivial palindrome} is a palindrome of length at most 1.
\end{definition}

\begin{definition}[Palindromic descriptor]\label{d:pald}
    A {\em maximal palindromic descriptor} is a pair of integers $(i, j)$ inside a palindromic fingerprint $F$ that describes a palindrome between indices $i$ and $j$.
\end{definition}
An example to a palindromic fingerprint can be found at~\cref{e:palf}.

\begin{definition}
Let $G=(V, E)$ be a graph. The function $\psi:V\to\mathbb{N}$ is a {\em graph coloring} if and only if $\forall (i, j) \in E, \psi(i) \neq \psi(j)$.
\end{definition}

\begin{definition}
A palindromic fingerprint $F$ is valid if and only if there exists a string $S$ such that $F=PF(S)$. Any string $S$ that satisfies $PF(S)=F$ is called a {\em preimage} of $F$. 
\end{definition}

\begin{example}
Let $S=(abc)^n$. There are no non-trivial palindromes in $S$, thus the fingerprint is $\{\}$.
\end{example}

\section{The Alphabet Upper Bound Theorem}\label{s:alph}

We are now ready to prove the main theorem, which we repeat below: 
\maintheorem*

\begin{example}
    Let $k=5$. By the theorem, $IPF(5)=9$, which means that the shortest fingerprint that cannot be reconstructed with less than 5 characters is of length 9. The fingerprint $A=(\{(2,8),(2,4),(6,8)\}, 9)$ can be described as a restriction graph that is a clique of size 5, hence $IPF(5)\le 9$. However, all fingerprints of length 8 can be recovered with four or fewer characters. \footnote{Will be proven later and can be checked manually.}
\end{example}

From now on, we represent the input as a pair: $(F, n)$, where $F$ is the palindromic fingerprint, and $n$ is the length of the string we aim to reconstruct.

The proof overview can be found at~\cref{a:pov}, and is a suggested reading.

\subsection{The Restriction Graph}
In this subsection we will present the {\em Restriction Graph} superficially. The restriction graph is an important tool to understand the palindromic fingerprint reconstruction in depth. The in-depth details of restriction graphs can be found at~\cref{ss:RG}.

\begin{definition}[Restriction graph]
The restriction graph of a palindromic fingerprint and length $(F,n)$ is a graph $G=(V,E)$ where there is a bijection between the reconstruction of $(F, n)$ and a coloring of $G$. 

Let $(F, n)$ be a palindromic fingerprint. We define $G=(V,E)$, where $V$ is a partition of $\{1,\ldots,n\}$, i.e.,
\begin{enumerate}
    \item $\forall v\in V,$ $v\subseteq \{1,\ldots, n\}$ and $v$ is not empty.
    \item $\underset{v\in V}{\bigcup v} = \{1, \ldots, n\}$
    \item $\forall v_1,v_2\in V, v_1\cap v_2 = \emptyset$
\end{enumerate}

Every node in $V$ presents symbols of $F$ that must be equal. Every edge in $E$ presents two different symbol groups in $V$ that must not be equal. Formally:

\begin{enumerate}
    \item $\forall i,j < n$, if there exists a palindromic restriction $F[i]=F[j]$, then $i,j$ are members of the same vertex in $G$.
    \item $\forall i,j < n$, where $i$ is in $v_i$ and $j$ is in $v_j$, if there exists a palindromic restriction $F[i]\ne F[j]$, then $(v_i,v_j)\in E$. 
\end{enumerate}

The graph has a self-loop, i.e., $(v_i,v_i)\in E$ for some $v_i$ if and only if the palindromic fingerprint is not reconstructible.
\end{definition}

\subsection{Proof Setup}
In this subsection, we refer to a palindromic fingerprint $F$ of length $n$ without declaring it. We include the more general definitions that will be used throughout both parts of the proof in this subsection.

\begin{definition}[Reconstruction degree]\label{d:rec-deg}
    The reconstruction degree $k$ of a fingerprint $F$ of length $n$ is the minimal number of characters required to reconstruct $(F, n)$, and is denoted by $\sigma(F, n)$, or simply $\sigma(F)$.
\end{definition}

\begin{definition}[Optimal fingerprint]\label{d:optf}
    A fingerprint $(F,n)$ is called optimal if $\forall F', n'<n$ $\sigma(F')<\sigma(F)$.
\end{definition}

\begin{definition}[Optimal string]\label{d:opts}
    A string $S$ is called optimal if $PF(S)$ cannot be reconstructed with fewer characters than $S$ has.
\end{definition}

\begin{definition}
The symbol $F[i]$ denotes the character at position $i$ in any string reconstructing $F$. Unlike standard strings, in a fingerprint, it is possible that neither $F[i]=F[j]$ nor $F[i]\ne F[j]$. For example, an empty palindromic fingerprint of length four can be reconstructed both as $abca$ and $abcd$. Therefore, $F[1]$ and $F[4]$ comparison is not defined by $F$. However, $F[1] \ne F[2]$ and $F[1] \ne F[3]$. We use standard string notations on palindromic fingerprints.
\end{definition}

\begin{definition}[Palindromic restrictions]
    Let $P$ be a palindromic descriptor $(i, j)$ inside $F$. The palindromic restrictions of $P$ are:
    \begin{enumerate}
        \item $\forall r$ where $0 \le r \le j-i$, $F[i+r]=F[j-r]$. 
        \item $F[i-1] \ne F[j+1]$.
    \end{enumerate}
\end{definition}

\begin{definition}
    Let $i,j$ be indices in $F$. We say that index $i$ {\em affects} index $j$ if and only $F[i]=F[j]$ or $F[i]\ne F[j]$. If $F[i]=F[j]$ and $F[j]=F[k]$, index $i$ does not necessarily affect index $k$.
\end{definition}

\begin{observation}[Adjacent relations]\label{o:adjr}
    For every index $i\le n+2$, $F[i]$ affects both $F[i+1]$ and $F[i+2]$.
\end{observation}

\begin{definition}
    Let $P$ be a palindromic descriptor in $F$. We define $start(P)=i,end(P)=j,center(P)=(i+j)/2$.
\end{definition}

\begin{definition}[Crossing palindromes]\label{d:crosspl}
    Let $P_1,P_2$ be palindromic descriptors in $F$. If $start(P_1)<start(P_2) \land end(P_1) < end(P_2)$, then $P_1,P_2$ are {\em crossing palindromes}.

We call a palindromic fingerprint with no crossing palindromes a {\em non-crossing palindromic fingerprint}.
An example can be found at~\cref{e:crosspl}.
\end{definition}

\begin{definition}
    Let $P_1,P_2$ be two crossing palindromes. We say $P_2$ is {\em dominated by} $P_1$ if the center of $P_2$ is part of $P_1$.
    
    Note that it is possible for $ P_1$ to dominate $P_2$ and $P_2$ to dominate $ P_1$ simultaneously.
\end{definition}

\begin{definition}
    Let $P_1,P_2,\ldots,P_n$ be a set of palindromes such that $P_i$ and $P_{i+1}$ are crossing, $i=1,\ldots,n-1$. We call  $F[start(P_1)..end(P_n)]$, the {\em substring defined by $P_1,P_2,\ldots,P_n$ on $F$} and denote it by $P_{1..n}=\bigcup_{i\le n} P_i$.
\end{definition}

\begin{definition}
    Let $P_1,P_2$ be two crossing palindromes. The {\em intersection between} $P_1,P_2$ is the substring $F[start(P_2)..end(P_1)]$ and denoted by $P_1 \cap P_2$.
\end{definition}


\begin{definition}[Standalone reconstruction]\label{d:sr}
    Let $I$ be a palindromic island $(j,\ell)$ in a fingerprint $(F,n)$. A {\em standalone reconstruction} of $I$ is the process of reconstructing $(F,n)$, where only restriction inside $I$ must be addressed in the reconstruction procedure.
\end{definition}
An example can be found at~\cref{e:sr}.

There are multiple ways to represent a palindrome, such as the start-length and center-radius notation. For ease of discourse, we use the start-length notation, i.e., $(i,\ell)$ stands for the palindrome $S[i..i+\ell-1]$.

\section{Proof for Non-crossing Palindromes}\label{s:pncp}

In this section, we will prove the following restricted lemma.
\begin{lemma}\label{l:m-pncp}
    Let $k\ge 3$ be an integer, and let $P_k$ be the {\em Zimin word}, i..e, $P_1=1$, $P_k=P_{k-1}kP_{k-1}$. 

    Considering only palindromic fingerprints that do not describe crossing palindromes, the palindromic fingerprint of $xP_{k-2}y$, where $x,y\notin P_{k-2}$ is the shortest palindromic fingerprint that can not be reconstructed with less than $k$ characters and is unique.
\end{lemma}



\subsection{Induction Hypothesis}
Let $F_k$ be the shortest palindromic fingerprint that cannot be reconstructed with fewer than $k$ characters, and let $S_k$ be a preimage with an alphabet of size $k$, i.e., $|\Sigma_{S_k}|=k$.

The structure of $S_k$ is defined as $xLtLy$, where $L$ is a palindrome, and  $x,t,y$ are letters not appearing in $L$. In addition, $xLt, tLy$ are the recursive instances $S_{k-1}$, for $k\ge 4$. We additionally claim that all strings reconstructing $F_k$ parameterize match to $S_k$.

The base of the induction is $S_3=123, S_4=12324$. Any palindromic fingerprint of length four can be reconstructed using three characters or less.


\begin{observation}
    Any empty palindromic fingerprint can be reconstructed with three characters - a prefix of $(abc)^\infty$. As a result, we can assume that for $k \ge 4$, and that we have a palindrome of length $\ge 2$. Moreover, any palindrome of length $\le 2$ consists of one character and, therefore, can be reconstructed with the string $(abc)^\infty$, where characters are doubled where a palindrome of length 2 occurs. Consequently, we will assume fingerprints to be of length at least four and contain a palindrome of length $\ge 3$.
\end{observation}

\subsection{Maximal Islands}\label{ss:maxisl}
We now define the {\em palindromic islands}, a key definition for the rest of the proof.

\begin{definition}[Palindromic island]\label{d:isl}
    Let $i\in \mathbb{N}$ be an index. Then the maximal palindromic island for index i, $I'_i=(j'_i,\ell'_i)$ is the longest maximal palindrome in $F$ such that $j'_i \le i < j'_i + \ell'_i$. If $i$ participates in no palindrome, then $I'_i=(i, 1)$. \\
    The palindromic islands of a string are the set $\mathcal{I}=\{I_1,I_2,\ldots\}$, where $\forall i \exists j$ s.t. $I'_i = I_j$.
\end{definition}
An example can be found at~\cref{e:isl}.

The name {\em islands} is used because of the absence of palindromes between two distinct islands. We will now state and prove this property.

\begin{lemma}\label{l:nocross}
    Let $I_i=(j_i,\ell_i),I_m=(j_m,\ell_m) \in \mathcal{I}$ be two distinct islands, $i<m$. $I_i,I_m$ do not share a palindrome, i.e., there does not exist a palindrome $(i,j)\in F$ such that $i< j_i+\ell_i, j \ge j_m$.
\end{lemma}
\begin{proof}
    If $i<m$ then $j_i < j_m$, as otherwise $I_m$ is contained in $I_i$, and is not an island. Moreover, because maximal palindromes are not allowed to cross each other, we know that $j_i+\ell_i \le j_m$. To show that $I_i$ and $I_m$ are not crossing, it is enough to show that there is no palindrome starting at the interval of $I_i$ and ending at the interval of $I_m$. However, such a palindrome would both cross $I_i$ and $I_m$, hence does not exist.
\end{proof}

\begin{corollary}[Maximal intervention]\label{o:maxint}
    Let $I=(j, \ell)$ be a palindromic island where $\ell \ge 2$. The island $I$ shares no palindromes with any other island, but the first and last two characters of the island can still be influenced by other islands. The influence is an inequality relation (\cref{o:adjr}).
\end{corollary}

\begin{observation}[Refined maximal intervention]
    Let $I_i=(j_i, \ell_i),I_{i+1}=(j_{i+1}, \ell_{i+1})$ be two neighboring islands. The only explicit relation between index $j_{i+1}+1$ and any index within $I_i$ is the inequality $S_k[j_{i+1}+1] \neq S_k[j_{i+1}-1]$. A similar claim is true for index $j_i+\ell_i-2$.
\end{observation}

\begin{definition}
    Let us consider any arbitrary algorithm $\mathcal{A}$ reconstructing palindromic fingerprints optimally, i.e., with the fewest possible characters. Let $I\in \mathcal{I}$ be an island, $I \in F$. $\Sigma_I$ is the set of characters used to reconstruct $I$ stand-alone using $\mathcal{A}$, and $\Sigma_S$ is the set of characters used to reconstruct the whole fingerprint.
\end{definition}

\begin{lemma}\label{lp:np4}\label{l:np2}
Let $m=\max(|\Sigma_{I_1}|,|\Sigma_{I_2}|,\ldots)$, then $|\Sigma_{S_k}| \le m + 2$.
\end{lemma}
\begin{proof}
    We will begin by proving that $|\Sigma_{S_k}| \le m + 4$.

    Let $I_i=(j_i,\ell_i)$ be the island with maximal $m=|\Sigma_{I_i}|$. To simplify the analysis, we initially assume that for all $i$, $\ell_i \ge 2$. We will address the general case later on.
    
    Due to the maximality of $I_i$, each of the other islands can be reconstructed stand-alone using at most $m$ distinct characters. At~\cref{o:maxint} we showed that the only indices sharing palindromic restrictions with $I_i$ are $j_i-1,j_i-2,j_i+\ell_i,j_i+\ell_i+1$. At~\cref{o:adjr} we claimed that the reconstruction of adjacent characters is predetermined. As a result, the reconstruction of $I_{i-1}$ and $I_{i+1}$ can still be performed using only $m$ characters, possibly introducing new characters on the margins.
    
    Consequently, reconstructing $I_{i+1},I_{i-1}$ requires at most $m$ distinct characters at each island, each using at most 2 new characters that weren't used in the reconstruction of $I_i$. 
    
    To reconstruct $I_{i+2}$, two additional characters not present in $\Sigma{I_{i+1}}$ might be required. However, assuming that the reconstruction of $I_{i+1}$ and $I_{i-1}$ introduced four new characters, it means that there are two characters in the reconstruction of $I_i$ that do not occur in the reconstruction of $I_{i+1}$ and on the last two indices of $I_i$, and we can use as the new characters required. It can be seen that even if $I_{i+1}$ introduced less than two new characters, an overall of 4 new characters will not be exceeded.

    We now consider islands of length one or trivial islands. In the first part of the proof, we labeled the islands individually. However, if there are islands of length one, we do not color them. Since the index of a trivial island has an explicit palindromic restriction with no more than four other indices, and because $|\Sigma_{S_k}| \ge 5$, we always have an existing color that the index can use as a valid label.

    Let us prove the tighter restriction $|\Sigma_{S_k}| \le m + 2$.

    Similar to the beginning of the proof, let $I_i$ be the heaviest island, describing $F[j_i..j_i+\ell_i-1]$, or simply $F[I_i]$. We now consider explicitly the two characters adjacent to $F[I_i]$ as $x_1,x_2,x_3,x_4$, where $x_1=F[j_i-2],x_2=F[j_i-1],x_3=F[j_i+\ell_i],x_4=F[j_i+\ell_i+1]$. 
    By the definition of islands, the maximal island $I_i$ describes a maximal palindrome, which implies $x_2\ne x_3$. As a consequence of $x_1,x_4$ having only one inequality relation with $I_i[1]=I_i[\ell_i]$, and since we assumed that both $x_2,x_3$ received new characters, we are left with two options:
\begin{enumerate}
    \item $x_1=x_2$, in which case a new character is not added to the reconstruction.
    \item $x_1\neq x_2$, in which case an existing character from $S_k$ can be chosen because the number of distinct characters in $T_i$ is at least 2.
\end{enumerate}
The same arguments can be made for $x_3,x_4$, proving we need no more than $m+2$ distinct characters for the entire string.
\end{proof}

In the proof of~\cref{l:np2}, we used a greedy labeling approach; such an approach is not yet proven to output the minimal alphabet size. This observation can help the reader understand why we do not specify a coloring algorithm but rather prove properties for all possible algorithms.

\begin{corollary}\label{c:nexs}
    The palindromic fingerprint $F_k$ is either of the form $P$, $xP$, or $xPy$, where $P$ is a descriptor of the maximal palindrome.
\end{corollary}

For the rest of the non-crossing proof, we will be using the symbol $P$ to identify the middle palindrome from~\cref{c:nexs}.

\begin{restatable}{lemma}{oddlemma}\label{l:odd}
The length of the palindrome $P$ is odd.
\end{restatable}
The proof is at~\cref{p:odd}.

\begin{lemma}\label{l:intpal}
    Let $P=LtR$ and assume the structure of $F_k$ is $xPy$, then $L=R$ are palindromes.
\end{lemma}
\begin{proof}
If $L$ is not a palindrome, then we can either choose $x=t$ or $t \in L$. \footnote{o/w, there cannot be a restriction $x \neq t$.}

As the structure $F_k$ has no crossing palindromes, its middle character $t$ does not participate in any palindromes other than $P$. Consequently, if we have the restriction $x\neq t$, it means $t$ is inside a palindrome bordering with $x$, which cannot be $L$ as it is not a palindrome. In both cases, the fingerprint of $xLRy$ cannot be reconstructed with fewer characters than $xLtRy$, but is shorter, contradicting that $F_k$ optimal (shortest possible).

Otherwise, we can set $x=t$ without violating any palindromic restriction. Again, the fingerprint of $tLtRy$ can be reduced to $tLRy$, and $tLRy$ cannot be reconstructed with less than $k$ characters.
\end{proof}

\begin{lemma}\label{l:uniqedeg}
    The string $S$ reconstructing $F$ is either of the form $P$, $xP$ or $xPy$, where $x,y \notin P$.
\end{lemma}
The proof is at~\cref{p:l:uniqedeg}.

\begin{lemma}\label{c:gstruc}
    An optimal string is of the form $xPy$ where $x,y\notin P$, and $P$ is a palindrome.
\end{lemma}
The proof is at~\cref{l:c:gstruc}.

\subsection{Recursive Structure}
We have seen at~\cref{c:gstruc} the structure of an optimal string. And now, we show that the middle palindrome $P$ is recursive.

\begin{lemma}[Explicit structure]\label{l:exps}
The string $S_k$ can be written as $xLtLy$, where $L$ is a palindrome, and $x,t,y$ are distinct characters that satisfy $x,t,y\notin L$.
\end{lemma}
\begin{proof}
    All of the essential pieces were proven before, and this proof will just glue them all together.
    First, by~\cref{c:gstruc} we know that $F_k=xPy$, and $x,y \notin P$. Using~\cref{l:odd} we know that there exists a middle character $t$. We proved at~\cref{l:intpal} that $L$ is a palindrome, and therefore the overall structure is $xLtLy$, and both $x,y$ are unique. However, if $t \in L$, then swapping $x$ and $t$ does not violate the fingerprint's restrictions - but contradicts~\cref{l:uniqedeg}.
\end{proof}
    
\begin{corollary}\label{l:mrecursive}
The string $S=S_k=xPy$ can be written as $xLtLy$ where $xLt$ and $tLy$ correspond to $S_{k-1}$.
\end{corollary}

We have thus shown a clear recursive structure that matches the original theorem, and the proof is completed for palindromic fingerprints with no crossing palindromes.

\section{Proof for General Fingerprints}

In this section, we will not forbid crossing palindromes from participating in the fingerprint. 

\begin{lemma}\label{l:m-pg}
    Let $k\ge 3$ be an integer, and let $P_k$ be the {\em Zimin word}, i..e, $P_1=1$, $P_k=P_{k-1}kP_{k-1}$. 
    
    The palindromic fingerprint of $xP_{k-2}y$, where $x,y\notin P_{k-2}$ is the shortest palindromic fingerprint that can not be reconstructed with less than $k$ characters and is unique.
\end{lemma}

\subsection{String Ordering}

Before beginning the proof, we will explain how a palindromic fingerprint can be consistently split into substrings.

\begin{definition}[Palindromic representative]
    Let $i$ be an index within $S$. The {\em representative palindrome} for index $i$ is the palindrome $P$ that:
    \begin{enumerate}
        \item Contains index $i$.
        \item Is not contained within any other palindrome.
        \item Has the minimal ending index. 
    \end{enumerate}
\end{definition}

\begin{lemma}
    All palindromes are either contained in or are palindromic representatives, and there are no two possible representatives for a single index.
\end{lemma}

\begin{proof}
    Let $P_1,P_2,\ldots,P_m$ be all palindromes that are not contained within any other palindromes, such that $start(P_i) < start(P_{i+1})$. For any palindrome $P=F[i..j]$, the palindrome $P$ contains index $j$, and has a minimal index $j$ of all palindromes that contain index $j$. According to our definition, $P$ is not contained in any other palindrome, hence a valid representative for index $j$. However, if there are two possible representatives for index $j$, it means that both end at index $j$ and are not of the same length, and therefore, one contains the other, contradicting the way $P$ was defined.
\end{proof}

\begin{observation}
    Any fingerprint $F$ can be split into a series of palindromic representatives.
\end{observation}

At~\cref{d:isl}, we have defined palindromic islands. The concept can be extended to {\em palindromic representatives}. We generalize an {\em island} to be a series of non-foreign palindromes. Every two palindromic islands are foreign.

\begin{lemma}\label{l:isl}
    Let $F$ be a palindromic fingerprint split into palindromic islands $I_1I_2\ldots I_m$. The minimal number of characters required to reconstruct $F$ is either:
    \begin{enumerate}
        \item $\max(\sigma(I_i))+2$ if $I_i$ is a palindrome, or
        \item $\max(\sigma(I_i))+1$ if $I_i$ is not a palindrome.
    \end{enumerate}
\end{lemma}

The proof is at~\cref{p:l:isl}.

\begin{observation}
    If the maximal degree island $I_i$ is not a palindrome, then the shortest palindromic fingerprint with reconstruction degree $\sigma(F)$ is of length at most $|I_i|+1$.
\end{observation}

\subsection{Proof}

The proof process is incremental. We first partition the island $I$ into crossing palindromes $P_1,P_2,\ldots,P_m$. We then prove many local lemmas about reconstructing consecutive palindromes from within $I$, and then logically "merging" consecutive crossing palindromes $P_i,P_{i+1}$ into one palindromic substring $P_{i..i+1}$, and showing that the lemmas also hold for these substrings. By showing that, we iteratively prove our lemmas on all $I$, as we can keep on merging new palindromes into the substring until it contains all $I$.

\begin{lemma}\label{l:contained}
    Let $P_1,P_2$ be crossing palindromes such that $P_1$ contains $P_2$. The fingerprint $F'=P_{1..2}$ satisfies $\sigma(F') \le \sigma(P_1)$.
\end{lemma}
The proof is at~\cref{p:l:contained}

The following lemma is a key lemma for the rest of the proof.

\begin{lemma}\label{l:add-one}
    Let $P_1,P_2$ be crossing palindromes where neither $P_1$ is contained in $P_2$ or the other way around, and let $F'=P_{1..2}$, then $\sigma(F')\le \max(\sigma(P_1),\sigma(P_2))+1$, and the number of distinct characters used to reconstruct $P_2$ and $P_1$ is at most $m$ in the reconstruction of $F'$.
\end{lemma}

\begin{proof}
    We assume w.l.o.g that $\sigma(P_1)\ge\sigma(P_2)$. Let's assume $P_1$ is reconstructed, and let $K$ denote the shared part of $P_1$ and $P_2$, i.e., $K=P_1\cap P_2$.

    Let $C_1,\ldots,C_m$ be all of the crossing palindromes between $K$ and $P_2$, such that for any $i,j\le m$, $C_i$ is not contained in $C_j$. First, we ignore all palindromes dominated by $P_1$, as they can not introduce any new characters to the reconstruction. Let $C_i$ be a palindrome that is not dominated by $P_1$. If $end(C_i)\ge center(P_2)$ then reconstructing $C_i$ suffices for the reconstruction of $P_2$. In this case, we can reconstruct $C_i$ instead of $P_2$ in the lemma, and all characters in $P_2$ will be reconstructed thereafter. Thus, let us assume that such a descriptor does not exist and choose two crossing palindromes $C_i,C_j$ ($start(C_i)<start(C_j)$).

    As we required before, both $center(C_i)>end(P_1),center(C_j)>end(P_1)$. Also, we required that $C_i$ is not contained in $C_j$ and vice versa, therefore $start(C_i)<start(C_j)$ and $end(C_i)<end(C_j)$. However, since $start(C_j)<end(P_1)$ and $center(C_i)>end(P_1)$, we have $start(C_j)<center(C_i)$, and $C_i$ is dominated by $C_j$. Let us call the crossing palindrome with maximal end index $C$. After $C$ is reconstructed, all other crossing palindromes $C_i$ do not have any free symbols. The only index in $C$ that is not reconstructed and has inequality palindromic relations with $P_1$ is $end(P_1)+1$.

    We now consider the reconstruction degree of $C$. We know that $\sigma(C)\le \sigma(P_2)\le m$. If $C$ introduces a new character at $end(P_1)+1=end(K)+1$, then $C$ has inequality relations with $m$ different characters. In palindromic fingerprints, inequality relations can only result from palindromes. Therefore, $m-1$ of the inequality relations occur within a palindrome. If $K$ is not this palindrome, then $K$ contains this palindrome, and the two characters next to it introduce two new characters, which implies that $m+1\le \sigma(P_2) \le \sigma(P_1) = m$, and that's a contradiction. Therefore, in order for $C$ to introduce a new character, $K$ must be a palindrome of reconstruction degree $m-1$. 

    We now consider $P_2$ again. We showed that $P_2$ has no crossing palindromes with $K$ except for the palindrome $C$ (if it exists). For the reconstruction degree of $P_2$ to be $m+1$, we need a symbol with $m$ inequality relations. If $C$ does not exist, we showed that there are at most $m$ distinct symbols in $P_2$ in the latter paragraph. If $C$ exists, then the index $end(C)+1$ might have many inequality relations with elements in $C$, and the index $end(P_1)+1$ is the new character $\$$ introduced at $C$. For $P_2$ to have an inequality relation with the new character at $end(P_1)+1$, there must exist a maximal palindrome bordering with $\$$. Still, if this maximal palindrome has $m-1$ characters and all have an inequality relation with the bordering characters, then the reconstruction degree of $P_2$ is $m+1$, which is a contradiction.
\end{proof}

\begin{lemma}\label{l:add-mid}
Let $P_1,P_2,P_3$ be crossing palindromes, and let $\sigma(P_1) \ge \max(\sigma(P_2),\sigma(P_3))$, then $\sigma(P_{1..3})\le \sigma(P_1) + 1$.
\end{lemma}
The proof is at~\cref{p:l:add-mid}.

\begin{corollary}\label{c:add-two}
    Let $P_1,P_2,\ldots,P_n$ be crossing palindromes.\\ $\sigma\left(\underset{i\le n}{\bigcup}P_i\right)\le \max\left(\sigma(P_i)\right) + 2$
\end{corollary}

\begin{lemma}\label{l:red-fin}
    If $\sigma(P_1 \cup P_2 \cup P_3) = \sigma(P_2) + 2$, the there exists a fingerprint $F'$ of length $|P_2|+2$ satisfying $\sigma(P_2)=\sigma(F')$.
\end{lemma}

\begin{proof}
    At~\cref{c:add-two}, we showed that the reconstruction of a palindromic island requires at most $\sigma(P_1 \cup P_2 \cup P_3) + 2$ characters. Moreover, at~\cref{l:add-one}, we showed that new characters can only be introduced on indices adjacent to the heaviest palindrome $P_2$. As a consequence, it implies that if the indices adjacent to $P_2$ could have been reconstructed from the existing characters, the reconstruction degree would be lower. Consequently, they must have been introduced as new characters. Therefore, the heaviest palindrome with its two adjacent indices cannot be reconstructed with less than $\sigma(P_2)+2$ characters. 
\end{proof}

\begin{corollary}
    The number of characters required to reconstruct a fingerprint $F$ is the number of characters required to reconstruct its heaviest palindrome plus, at most, two.
\end{corollary}

The last corollary concludes the vital components of the proof, as it shows that any fingerprint can be reduced to a fingerprint with no crossing palindromes. The methodology of converting such a string without reducing the reconstruction degree or increasing the length of the fingerprint is as follows:
\begin{enumerate}
    \item At~\cref{l:isl} we show that the whole fingerprint is dominated by the heaviest island, so the fingerprint can be reduced to one island, with up to two additional characters.
    \item At~\cref{c:add-two} we show that a single island is dominated by the heaviest representative, so we can only keep the representative.
    \item Use this method recursively on the last representative, omitting the two additional characters added in step (1). Eventually, we will conclude that the optimal string has no crossing palindromes, as the crossing palindromes can always be removed, as shown at~\cref{l:red-fin}.
\end{enumerate}

And since the overall fingerprint has no crossing palindromes, we achieve the relaxed requirement of~\cref{l:m-pncp}, and the proof of~\cref{t:main} is completed \qed.

\section{Properties of the Results}
In this section, we present the general properties of our result, show that the algorithm presented by I et al. indeed outputs a globally minimal alphabet, and present applications on real-world data.

\begin{lemma}\label{l:clique}
Let $F_k$ be a minimal-length palindromic fingerprint that cannot be recovered with less than $k$ characters. The restriction graph of $F_k$ is a clique.
\end{lemma}
The proof is at~\cref{p:l:clique}.

\begin{lemma}
    The palindromic fingerprint reconstruction algorithm of~\cite{tomohiro:10} produces a minimal alphabet.
\end{lemma}

\begin{proof}
    At~\cref{l:add-one} we analyzed the positions that can introduce new characters to a reconstruction. We showed that these positions only occur on the edge of other palindromes, and we showed they cannot increase the overall reconstruction degree of a single palindrome. Therefore, instead of analyzing the reconstruction algorithm as if it started from the heaviest palindrome, the analysis can start from the first possible index, and~\cref{l:add-one} will still hold, and since the reconstruction algorithm~\cite{tomohiro:10} reconstructs the string from the first index, and since it only introduces new characters when necessary, its settings are the same as~\cref{l:add-one} and it outputs a minimal alphabet.
\end{proof}

\begin{lemma}[Open Problem by~\cite{tomohiro:10}]
    Finding a string that has a given set of maximal palindromes $(F,n)$ and contains exactly $k$ characters, where $k$ is a predefined parameter can be achieved in time proportional to the construction of the restriction graph.
\end{lemma}
\begin{proof}
    Let us build the restriction graph $G$ of the input fingerprint $(F,n)$ and run the algorithm by~\cite{tomohiro:10} to achieve a coloring of $G$ using a minimal number of colors. The maximal number of characters in the resulting string $S$ is the number of vertices in $G$. The minimal number is the number of distinct characters returned by the algorithm. Any number not in this range is not achievable. To increase the number from the minimal number to the maximal, we can take two vertices that were colored with the same color, and color one of them with a new color. When all vertices have a distinct color, the maximal number is reached.
\end{proof}

\begin{corollary}
    Let $(F,n)$ be an arbitrary palindromic fingerprint. Let $G$ be the restriction graph of $F$. The chromatic number of the graph $\chi(G)$ satisfies $\chi(G) \le \log(n-1)+2$.
\end{corollary}

\begin{corollary}
    For every string $S$ over finite alphabet $\Sigma$, $k=|\Sigma|$, $PF(S)$ cannot have $PF(S_{k+1})$ as a subset with regard to any offset.
\end{corollary}

More specifically, we know that $S_5=512131214$ is only 9 characters long, but its fingerprint cannot appear in any fingerprint describing a DNA sequence, at any offset.

\section{Conclusion and Future Work}
We have defined a graph representation of palindromic fingerprints, the restriction graph, where the chromatic number of the graph equals the alphabet size of the preimage of a given palindromic fingerprint.

We have also defined the combinatorial structure of the shortest preimage for a valid palindromic fingerprint over an alphabet of size $k$.

While this study has provided valuable insights into palindromic fingerprints, several avenues remain unexplored. Future work on the topic includes:
\begin{enumerate}
    \item Finding an efficient algorithm to construct the restriction graph.
    \item Exploring the minimal presentation of a palindromic fingerprint as a method of data compression.
\end{enumerate}

\bibliographystyle{plainurl}
\bibliography{paper} 

\newpage
\appendix

\section{Graphs and Reconstructions}
\begin{claim}\label{c:colnp}
    Reconstructing a general string with positive and negative constraints over a minimal alphabet is NP-hard, even when negative constraints are limited to the form $A[i]\neq A[j]$.
\end{claim}
\begin{proof}
    Let us show a reduction from graph coloring: Given a graph $G=(V,E)$, for every $(v_i,v_j)\in E$, we add the restriction $T[i] \neq T[j]$. We add no positive restrictions to the set of restrictions. 

    Any string $S$ satisfying our restrictions is a valid coloring, with $v_i$ being colored with $S[i]$, and we never have two neighboring vertices with the same color, because we have the restriction $S[i]\neq S[j]$ for every edge $(v_i,v_j)$. For similar reasons, every coloring of $G$ is a possible reconstruction of the given restriction set. Hence, the alphabet size of the reconstructed string is equal to the number of colors, and hence, reconstructing using a minimal alphabet is equivalent to coloring with a minimal number of colors and is NP-hard.
\end{proof}

\subsection{Palindromic Fingerprint Reconstruction and the Restriction Graph}\label{ss:RG}

The {\em restriction graph} $G$ of a palindromic fingerprint $(F,n)$, is a graph whose coloring is bijective to a reconstruction of $(F,n)$. \\
The equivalence between restriction graph coloring and palindromic fingerprint reconstruction gives us trivial bounds on the number of labels required to reconstruct a given palindromic fingerprint - the lowest possible number is the chromatic number of the graph $\chi(G)$, and the highest possible number is the number of nodes in $G$, $|V|$.

Before starting, we remind the reader that Manacher's algorithm \cite{man:75} finds all maximal palindromes in a string in linear time, hence converting a string $S$ to its corresponding palindromic fingerprint $PF(S)$ is done in linear time.


Let $(i,j) \in F$. The pair $(i,j)$ imposes a set of equality and inequality constraints that a source text $S$ must satisfy. Namely, for $S[i\ldots j]$ to be a palindrome, we must have $S[i+r] = S[j-r]$ for every $r\in \{0\ldots (j-i)\}$. Furthermore, for $S[i\ldots j]$ to be a maximal palindrome, we must have $S[i-1] \neq S[j+1]$ (if these indices are well defined). We say that these constraints are {\em derived} from the pair $(i,j)$. We say that a constraint is derived from $F$ if it is derived from one of the pairs in $F$.

\begin{definition}[Equality Graph]\label{d:eqg}
The equality graph of a fingerprint $(F,n)$ denoted as $G_{=}(F,n)=(V, E)$ is an undirected graph defined as follows:\\
$V=\{i\}_{i=1}^n$, and $(i, j) \in E \iff S[i] = S[j]$ is a constraint derived from $F$, where $1\leq i<j\leq n$. 
\end{definition}

The following directly follows from the definition of $G_{=}(F,n)$.
\begin{observation}
Let $A_1,A_2, \ldots ,A_\ell$ to be the connected components of $G_{=}(F,n)$. Every source text $S$ with $PF(S) = F$ must have $S[i] = S[j]$ if $i,j \in A_k$, for any $1\leq k \leq\ell$.
\end{observation}

We now define the restriction graph of $(F,n)$.
\begin{definition}[Restriction Graph]\label{d:rsg}
Let $G'$ be the restriction graph of a palindromic fingerprint $A$, i.e., $G' = G_{=}(F,n) = (V', E')$. The \textrm{restriction graph} $G=(V, E)$ is defined as follows:
\begin{itemize}
    \item $V=\{A_k\}_{k=1}^\ell$, where $A_k$ is a connected component of $G'$.
    \item $E=\{(A_{k_1}, A_{k_2}) ~|~ \exists i \in A_{k_1} \exists j \in A_{k_2} ~ s.t. ~ (i+1,j-1)\in F\}$. In words, it means that index $i$ cannot have the same symbol as index $j$ because of a palindromic restriction.
\end{itemize}
\end{definition}

\begin{observation}
    Let $G=(V,E)$ be a restriction graph for a palindromic fingerprint $(F,n)$. The vertices $V=(v_1,v_2,\ldots)$ form a partition of $[n]=\{1,2,\ldots,n\}$, i.e., $\forall v_i, v_i\subseteq [n]$, $\bigcup V = [n], \forall v_i\neq v_j,\text{ } v_i \cap v_j = \emptyset$.
\end{observation}

\begin{observation}\label{o:neqres}
If there is an edge $(A_{k_1}, A_{k_2})$, then every string $S$ reconstructing $(F,n)$ satisfies $\forall i\in A_{k_1}, \forall j\in A_{k_2}, S[i] \neq S[j]$.
\end{observation}

From ~\cref{d:rsg}, there is an edge between each pair of connected components which has an inequality constraint from $F$. The next step is to color the restriction graph such that there are no connected vertices colored in the same color.


\begin{lemma}\label{l:ctos}
Let $\psi$ be a coloring of the restriction graph $G$ on input $(F, n)$, one can reconstruct a string $S=\psi'(1)\psi'(2)\ldots \psi'(n)$, where $\psi'(i)=\psi(A_k) \iff i \in A_k$, and $A_k$ is a vertex of $G$.
Then, the translation from $G$ to $S$ is a bijection. 
\end{lemma}
\begin{proof}
Converting a coloring to a parameterized string is straightforward. Let $\psi$ be the coloring of $G$. Define $\psi':V\to \mathbb{N}$ as $\forall v' \in V'~\forall v \in v', \psi'(v)=\psi(v')$. Now, $S=\psi(1)\psi(2)\ldots \psi(n)$ is a string, satisfying all the restrictions of $(F, n)$. $V'$ is the vertex set of the equality graph, i.e., each vertex represents an index. 

Converting a parameterized string to a coloring can be done similarly. Let $S$ be a preimage of $(F,n)$. $\forall A \in V', i \in A, \psi(i)=S[i]$. The previous definition is well defined, since by the construction of the graph, $\forall A \in V', \forall i, j \in A, S[i]=S[j]$. $\psi$ is a valid coloring, as shown by observation~\ref{o:neqres}.
\end{proof}

\begin{corollary}
\cref{l:ctos} implies that there is at least one preimage for a given palindromic fingerprint if and only if the restriction graph does not have self-loops.
\end{corollary}

A detailed example can be found at~\cref{e:rsg}.





\section{Proof Supplementary Components}

\subsection{Proof Overview}\label{a:pov}
This section highlights the proof for~\cref{t:main}. 

The proof is divided into two parts. We first prove that a fingerprint that allows no {\em crossing palindromes} satisfies our requirements. Based on that lemma, we prove that the theorem is also true for a general string. 

To prove the lemma on restricted fingerprints that do not allow crossing palindromes, we use an inductive approach where we establish that an optimal string is formed by combining two smaller optimal strings. The form we assume for strings of length $\ge 3$ is $xLtLy$, where $L$ is a palindrome, all $x,t,y$ are unique characters, i.e., characters that do not occur in $L$, and both $xLt$ and $tLy$ are of the same recursive form.

In the first part of the proof, we introduce {\em palindromic islands}. A palindromic island is a palindrome not contained in a longer palindrome. Using the new definition, we prove a key lemma - that every two islands can be reconstructed almost independently, that is to say, the number of characters required to reconstruct a string is the number of characters required to reconstruct its {\em heaviest} island plus up to two additional characters, where heaviest refers to the island that requires more characters to reconstruct. Using the previous key observation, we analyze these two additional characters and prove that they must be introduced on the indices adjacent to the previously mentioned island.

After proving that the overall number of characters in the reconstructed string is the same as the number of characters in its heaviest island and the two surrounding indices, we conclude that an optimal fingerprint (as defined at~\cref{d:optf}) must contain at most one island of length $>1$. \\
Knowing that there can only be one island of length $>1$, we show that there must be two other islands of length one and that both of these new islands present a new character to the reconstructed string. The overall structure of an optimal fingerprint at this point is $xPy$, where $x\neq y, x,y\notin P$, and $P$ is a palindrome. \\
Lastly, we show that $P$ can be rewritten as $P'tP'$, where $P'$ is a palindrome and $t \notin P'$, supporting the induction hypothesis. 

We conclude the first part of the proof by showing that the reconstructed string has a unique fingerprint and that its restriction graph is a clique and, hence, has a unique parameterized reconstruction.

The second part of the proof generalizes the theorem to a general string. Similar to the first proof, we will subdivide the string into islands and the islands into locally maximal palindromes. Almost identical to the first part, we will show that different islands can be reconstructed almost independently. Later, We will discuss the prerequisites of a palindrome to introduce a new character to the reconstruction and will show that some crossing palindromes don't change the overall reconstruction degree. The last part shows that reconstructing two adjacent crossing palindromes does not introduce many more characters than reconstructing the heavier of the two and will generalize it to any sequence of crossing palindromes.

In conclusion, we will show that the number of characters required to reconstruct a general string is very close to the number of characters necessary to reconstruct the heaviest palindrome within, eventually reducing the problem to that of a string with no crossing palindromes and concluding the proof.

Some of the lemmas are rather technical and are left in the appendix.

Before starting the proof, we will introduce some definitions that will be used throughout the proof.

\subsection{Technical Proofs}

\begin{proofcount}[Proof for~\cref{l:odd}]\label{p:odd}
Assume that $|P|$ is even. Denote the even-sized palindromic descriptor $P$ as $LttR$, where $L,R,t$ are symbols representing indices of $P$.

Let $P'$ be the longest maximal sub-palindrome of $P$ that contains the middle symbol $t$. $P'$ might be of length one (i.e., containing only $t$). Let us rewrite $P=L'P'P'R'$. We can reduce $P$ to $P''=L'P'R'$, and $P''$ has an isomorphic restriction graph - all inequalities are kept, as $L$ can only share palindromic restrictions with the first instance of $P'$, and $R$ with the second instance. As a result, squashing them into one instance can not lift any labeling restriction on either $L,P'$ or $R$.

Also, there cannot be a palindrome $L'pref(P')$ where $|pref(P')|>|P'|/2$, as it would imply that $P'$ has an internal crossing palindrome. Because the longest possible palindrome prefix is at most $|P''|/2$, it is kept as a non-crossing palindrome after the reduction.

Thus, we found $P''$ with the same restriction graph; therefore, $PF(P)$ is not optimal. 
\end{proofcount}

\begin{proofcount}[Proof for~\cref{l:uniqedeg}]\label{p:l:uniqedeg}
    Because of the maximality of $P$, $x\neq y$.

    If $x\in P$, then the form $xP$ is not optimal, as it could be reduced by one character to the string $P$ without losing any restrictions. The same can be claimed for $xPy$, where $x,y \in P$. \\
    We now look at the case where $x \in P$ but $y \notin P$.

    We observe that $x$ is the only possible character in $P$ that could have been chosen because if there were two options, we could use the second for $y$.

    By~\cref{l:intpal}, we claim that $L=R$ must be palindromes. The number of distinct characters labeled for $xLt$ cannot go below $k-1$, but we assumed that $F_{k-1}$ is both optimal and unique, hence $xLt$ is either longer than $S_{k-1}$ or equal to it, which implies $x,t\notin L$, and that is a contradiction.
\end{proofcount}

\begin{proofcount}[Proof for~\cref{c:gstruc}]\label{l:c:gstruc}
    As we have seen at~\cref{l:np2}, when $|\Sigma|\ge 5$, the optimal string can be represented with one island and at most one adjacent character on each side, leaving us with the following three possible structures to examine: \\
    $xPy,xP$ and $P$. In all cases, $x,y\notin P$, due to~\cref{l:uniqedeg}. 

    Let us assume $xP$ is possibly the optimal form. Since $x\notin P$ and because the reconstruction is optimal, we learn that we could not have chosen $x$ to be equal to any character $c\in P$. Rewrite $P=LtR$, we know that all palindromic inequality restrictions on $x$ can be derived from $Lt$ as well, which disqualifies $xP$ from being optimal. \\
    Similarly, let us assume $P$ is the optimal form. Let us split $P=LtR$, due to the palindromic fingerprint not having crossing palindromes, rewriting $P'=Lt$ will not change the restrictions for any character in $Lt$, hence, $P$ is not optimal.
\end{proofcount}

\begin{proofcount}[Proof for~\cref{l:isl}]\label{p:l:isl}
    We already proved at~\cref{l:np2} that the statement holds if $I_i$ is a palindrome. So, we only consider the case where $I_i$ is not a palindrome. In this case, $I_i$ is of length at least four and contains at least two different characters, as if it only had one character, it was a palindrome, and if it was of length three or less, then the two crossing palindromes are of length 2, which can only occur if the string is $aaa$. 
    
    After showing that $I_i$ is of length four or more and contains at least two different characters, let us assume to the contrary that $\sigma(I_i)+2$ characters are required to reconstruct $PF$. As we showed at~\cref{l:np2}, only four characters have inequality relations with $I_i$ - the two adjacent characters on each side of the island. The characters in positions $end(I_i)+2,start(I_i)-2$ (if exist) only have an inequality relation with the last and first characters of $I_i$ accordingly, hence can be chosen for the other character in $I_i$. However, the characters next to $I_i$ can have an unbounded number of inequality restrictions. Let us assume a new character $\$$ is added to the left of $I_i$. Since $I_i$ is not a palindrome, the character on the right can also be $\$$, as $\$$ is a new character that cannot violate any palindromic restriction within $I_i$. Since $I_i$ is not a palindrome, it neither violates a restriction with the previously assigned new character $\$$. Neither equals any further characters, as they were chosen as characters inside $I_i$.
\end{proofcount}

\begin{proofcount}[Proof for~\cref{l:contained}]\label{p:l:contained}
    By definition of crossing palindromes, $start(P_1)<start(P_2)$ and $end(P_1)<end(P_2)$. Therefore, the inequalities defined by the palindromes $P_1,P_2$ do not interfere with the reconstruction of $P_1$ or $P_2$. Also, by definition of contained palindromes, after reconstructing $P_1$, we already reconstructed all characters in $P_2$, but no new character inequality restriction was introduced.
\end{proofcount}

\begin{proofcount}[Proof for~\cref{l:add-mid}]\label{p:l:add-mid}
    To the contrary - we should reconstruct $P_1$ first, by~\cref{l:add-one}, $P_2$ can only introduce one new character to the reconstruction. At~\cref{l:add-one}, we showed that $P_2$ has at most one new character and at most $\sigma(P_1)$ distinct characters overall. We now need to show that $P_3$ does not introduce another character. Because $P_2$ is ordered before $P_3$, we know that $start(P_2)<start(P_3)$ and $end(P_2)<end(P_3)$. Therefore, all characters common to $P_1$ and $P_3$ are also common to $P_1$ and $P_2$. We showed that either $P_2$ introduced no new characters, or a character exists in $P_1$ but not in $P_2$. If $P_2$ introduced no new characters, then the previous lemma applies directly to $P_3$. Otherwise, there exist two characters $\$_1,\$_2$ in $P_1$ and $P_2$ (resp.) that are unique. The character $\$_1$ does not have any palindromic inequality with characters in $P_3$ that are not reconstructed, therefore whenever a new character needs to be introduced, $\$_1$ can be chosen. We previously showed that no more than one character will be needed, keeping an overall of a single added character.
\end{proofcount}

\begin{proofcount}[Proof for~\cref{l:clique}]\label{p:l:clique}
    The proof is simple and inductive. The induction claim is that the restriction graph of $F_k$ is a clique. Note that the restriction graph of a string and the graph of its fingerprint produce the same graph.
    
    We know that $S_3=123$, and by~\cref{o:adjr}, the restriction graph of $S_3$ is a clique.

    We now prove for $k+1$: By the recursive structure of $F_{k+1}$, we know that if $S_k=xP_kt$, then $S_{k+1}=xP_ktP_ky$. By the induction claim, $S_k=xP_kt$ has a clique restriction graph. Moreover, because of the recursive structure, we know that $y$ must be connected to all nodes in the restriction graph representing $P_k$ and $t$. What is left is to show that $x$ and $y$ are connected, which is trivial, as $P_ktP_k$ is a palindrome.
\end{proofcount}
\section{Supplementary Examples}
\begin{example}[Crossing Palindromes - \cref{d:crosspl}]\label{e:crosspl}
    Let $S=1213431211$ be a string. The non-trivial palindromes in the string are $(1,3),(1,9),(7,9),(9,10)$. By the definition of crossing palindromes, $(1,3)$ and $(1,9)$ do not cross because the palindromes share the same starting point. For similar reasons, $(1,9)$ and $(7,9)$ are not crossing. However, the palindromes $(i_1=1,i_2=9),(i_1=9,i_2=10)$ are crossing, because $i_1<i_2 \land i_2\le j_1 < j_2$.
\end{example}

\begin{example}[Main problem - Problem~\ref{p:main}]\label{e:p1}
    Let $F=(\{\}, 10)$ be a palindromic fingerprint. A possible reconstruction of $F$ is the string $S_1=0123456789$. Another possible reconstruction is $S_2=0120120120$.

    Let $F=(\{(2,4)\}, 5)$ be a palindromic fingerprint. A possible reconstruction for this fingerprint is $S=01213$. Later in the paper, we prove that all possible reconstructions for $F$ parameterize match to $S$.

    Let $F=(\{(1,4),(2,4),(2,5)\}, 5)$ be a palindromic fingerprint. No string $S$ exists such that $PF(S)=F$. Proof: 
    \begin{equation*}
    \begin{split}
        (1,4)=>S[1]=S[4],S[2]=S[3]\\ (2,4)=>S[2]=S[4] \\ (2,5) => S[2]=S[5],S[3]=S[4]
    \end{split}
    \end{equation*}
    hence, $S[1]=S[5]$, but by maximality of $(2,4)$, $S[1] \neq S[5]$.
\end{example}

\begin{example}[Palindromic fingerprint - \cref{d:palf}]\label{e:palf}
Let $S=abacabada$. The maximal palindromes of $S$ are
$$
\{(1, 7), (1, 3),(5,7), (7, 9)\} \cup $$
$$ \{(1,1), (3,3), (5,5), (7,7), (9,9)\}
$$
and thus the palindromic fingerprint is $PF(S) = \{(1, 7), (1, 3),(5,7), (7, 9)\}$. Palindromes are visualized at~\cref{f:pal}

\begin{figure}[ht]
    \caption{The palindromes in $S$ }
    \includegraphics[width=140pt]{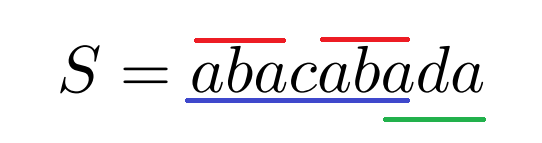}
    \label{f:pal}
\end{figure}

The example above presents the fingerprint as a union of two sets. The set on the bottom of the equation consists of trivial maximal palindromes - palindromes of length one \footnote{Some authors refer to empty palindromes as well, in this paper we do not regard them at all.}, hence we ignore them in the fingerprint. The set on top consists of the rest of the maximal palindromes. Non-maximal palindromes, like $(2,2)$, were omitted.

It is worthwhile to mention that either $(1, 3)$ or $(5, 7)$ could have been omitted because each, in addition to the $(1, 7)$ palindrome, implies the other. However, we do not omit such palindromes from the fingerprint. Such palindromes are called {\em implied maximal palindromes}.

An example of a valid palindrome that must not participate in the fingerprint is $(2, 6)$, because it is not maximal.

\end{example}

\begin{example}[Restriction graph - \cref{d:rsg}]\label{e:rsg}
    Consider the string $S=41213121566757$.

    The fingerprint of the string is $PF(S)=\{(2,4),(2,8),(6,8),(10,11),(12,14)\}$, which is all the maximal palindromes. 

    The equality graph, $G_{=}$, is a graph where each index in the original string is assigned to a node. In the graph, we only connect nodes that ought to be connected by a palindromic restriction. The equality graph is illustrated at~\cref{f:eqg}.

    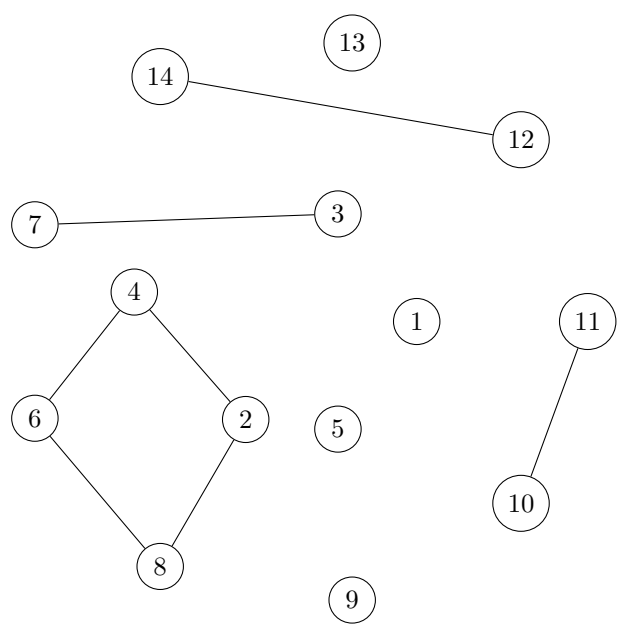
\begin{figure}
        \centering
        \begin{tikzpicture}[scale=1.5]

            \foreach \i in {1,5}{
                \node[circle, draw] (\i) at ({72*(\i - 1)}:1) {\i};
            }

            \foreach \i in {8,...,14}{
                \node[circle, draw] (\i) at ({40*(\i-2)}:2.5) {\i};
            }

            \node[circle, draw] (2) at ({80*(4-1)}:1) {2};
            \node[circle, draw] (3) at ({72*(2-1)}:1) {3};
            \node[circle, draw] (4) at ({85*(3-1)}:1.5) {4};
            \node[circle, draw] (6) at ({40*(7-2)}:2.5) {6};
            \node[circle, draw] (7) at ({40*(6-2)}:2.5) {7};

            \draw (3) -- (7);
            \draw (12) -- (14);
            
            \draw (2) -- (4);
            \draw (4) -- (6);
            \draw (2) -- (8);
            \draw (6) -- (8);

            \draw (10) -- (11);
          \end{tikzpicture}
          \caption{The equality graph of $S$.}
    \label{f:eqg}
    \end{figure}

    The restriction graph $G$ shrinks every connected component to only one node. We labeled the nodes with their original indices, comma separated. In the restriction graph, we draw an edge between two nodes that should be reconstructed using another label. Note that although we omitted trivial palindromes in the fingerprint, we still know they are there, and when considering the restriction graph we also consider the restrictions of trivial palindromes. For example, we can see that there is no palindrome centered at 10, which means that the palindrome centered at ten is trivial, hence $S[9] \neq S[11]$. The restriction graph is at~\cref{f:restg}.

    Our last step towards reconstruction is coloring the graph. We can see that the subgraph with nodes $\{(1), (2,4,6,8), (3,7), (5), (9)\}$ is the complete graph $K_5$, hence there is no coloring with less than five colors. Also, there are only eight nodes, implying no coloring with more than eight colors exists.

    Our original string $S$, is a valid coloring using 7 colors. The string $S=41213121566787$ is the naïve coloring that assigns each node with another color. The string $S=45253525133212$ represents an optimal coloring of the graph, i.e., coloring with the fewest possible colors. The coloring is presented at~\cref{f:optc}.

\begin{figure}
    \begin{subfigure}{0.45\textwidth}
        \centering
        \begin{tikzpicture}[scale=1.5]
        
            \node[circle, draw] (1) at (-1.5,-1.5) {$v_1$};
            \node[circle, draw] (2) at (-2,0) {$v_2$};
            \node[circle, draw] (3) at (-1,2) {$v_3$};
            \node[circle, draw] (4) at (0,-2) {$v_4$};
            \node[circle, draw] (5) at (0,2) {$v_5$};
            \node[circle, draw] (6) at (2,-1.5) {$v_6$};
            \node[circle, draw] (7) at (2,0) {$v_7$};
            \node[circle, draw] (8) at (2,2) {$v_8$};
            
            \foreach \i in {2,4,5,7,8}{
                \foreach \j in {2,4,5,7,8}{
                    \ifnum\i<\j
                        \draw[line width=0.5mm, red] (\i) -- (\j);
                    \fi   
            }}

            \foreach \i in {1,3,6}{
                \foreach \j in {1,3,6}{
                    \ifnum\i<\j
                        \draw (\i) -- (\j);
                    \fi   
            }}            
            \draw (1) -- (2);
            \draw (2) -- (6);
            \draw (6) -- (7);
          \end{tikzpicture}
          \caption{The restriction graph of $S$, biggest clique is highlighted}
          \label{f:restg}
        \end{subfigure}
        \hfill        
        \begin{subfigure}{0.45\textwidth}
          
          \begin{tikzpicture}[scale=1.5, every node/.style={circle, draw, inner sep=1pt, line width=1.5pt}]
        
            \node[circle, draw, color=red] (1) at (-1.5,-1.5) {$v_1$};
            \node[circle, draw, color=orange] (2) at (-2,0) {$v_2$};
            \node[circle, draw, color=orange] (3) at (-1,2) {$v_3$};
            \node[circle, draw, color=blue] (4) at (0,-2) {$v_4$};
            \node[circle, draw, color=darkgray] (5) at (0,2) {$v_5$};
            \node[circle, draw, color=blue] (6) at (2,-1.5) {$v_6$};
            \node[circle, draw, color=violet] (7) at (2,0) {$v_7$};
            \node[circle, draw, color=red] (8) at (2,2) {$v_8$};
            
            \foreach \i in {2,4,5,7,8}{
                \foreach \j in {2,4,5,7,8}{
                    \ifnum\i<\j
                        \draw (\i) -- (\j);
                    \fi   
            }}

            \foreach \i in {1,3,6}{
                \foreach \j in {1,3,6}{
                    \ifnum\i<\j
                        \draw (\i) -- (\j);
                    \fi   
            }}            
            \draw (1) -- (2);
            \draw (2) -- (6);
            \draw (6) -- (7);
          \end{tikzpicture}
          \caption{Optimal coloring of $G$}
          \label{f:optc}
    \end{subfigure}
    \caption{$$v_1=\{9\},v_2=\{12,14\},v_3=\{13\},v_4=\{5\},v_5=\{1\}, v_6=\{10,11\},v_7=\{2,4,6,8\},v_8=\{3,7\}$$}
\end{figure}

\end{example}

\begin{example}[Palindromic islands - \cref{d:isl}]\label{e:isl}
    Let $S$ be a string, $S=eaabbaadeed$, with its palindromes highlighted. The islands are maximal palindromes not contained in any other palindromes, and they are highlighted in blue at~\cref{f:isl}.
    \begin{figure}[H]
        \caption{The palindromes and islands of $S$}
            \includegraphics[width=300pt]{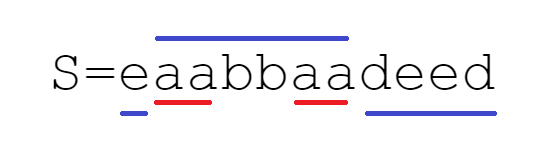}
        \label{f:isl}
    \end{figure}
    
\end{example}

\begin{example}[Standalone reconstruction - \cref{d:sr}]\label{e:sr}
Let $S$ be a string, $S=eaabbaadeed$. The palindromic islands of $S$ in start-length notations are $\{(1,1),(2,6),(8,4)\}$, or $I_1=(1,1),I_2=(2,6),I_3=(8,4)$. A standalone reconstruction of $I_1$ can be $aaaaaaaaaaa$, and can also be $a**********$. A standalone reconstruction of $I_2$ can be $aaabbaaaaaa$, and can also be $*ccaacc****$, but cannot be $eacbbcadeed$, because the violated palindromic restriction $(2,2)$ is within $I_2$.
\end{example}

Using a general reconstruction algorithm $A$, it is straightforward to independently reconstruct an island $I=(j,\ell)$. All that needs to be done is to reconstruct the fingerprint of $T[j..j+\ell-1]$, where $T$ represents the reconstructed original fingerprint. It is important to note that this explanation is informal and provided solely to give an intuitive understanding.

\end{document}